\documentclass[a4paper,10pt]{article}
\usepackage[utf8]{inputenc}
\usepackage{graphicx}
\usepackage{amsmath}
\usepackage{amssymb}
\usepackage{amsthm}
\newtheorem{prop}{Proposition}
\newtheorem{theo}{Theorem}

\title{Bulk viscosity in 2D relativistic fluids: the effects of temperature and modifications to the Rayleigh-Brillouin spectrum}
\author{A L Garc\'ia- Perciante, L Franco-P\'erez and A R M\'endez}

\begin{document}

\maketitle
\begin{abstract}
The dependence of the bulk viscosity with the relativistic parameter
$z=kT/mc^{2}$, obtained through the complete Boltzmann equation \cite{Ch-E,CercignaniKremer},
is thoroughly analyzed. A complete and rigorous examination of the
relevant non-relativistic and ultra-relativistic limits is carried
out in the case of a hard disk model and compared with the results
obtained in the relaxation time approximation. The modifications of
the non-vanishing bulk viscosity to the Rayleigh-Brillouin is briefly
discussed.
\end{abstract}

\section{Introduction}

\label{intro}
The study of relativistic bidimensional fluids is a relevant and interesting topic
which is still in need of extensive analysis. Additional to its importance
in the modelling of axisymmetric systems, for example a charged fluid
in the presence of a magnetic field, orbiting gases, or accretion
in gravitational potentials, the interest in the study of bidimensional
systems in relativistic scenarios has seen a substantial increase
due to the new generation of thin materials and their several applications.
Also, relativistic fluids still pose a challenge, plagued with unaswered
questions and topics of intense debate even in very fundamental issues,
numerical simulations in low dimensionality have been one of the most
valuable assets in order to corroborate theoretical predictions.

The establishment of relativistic transport coefficients in a special
relativistic framework from the complete Boltzmann equation has only
been addressed recently in the 2D scenario. Indeed,
in a separate publication (see Ref. \cite{Garcia-Mendez2018}) the
constitutive equations are established and the transport coefficients
expressed in terms of collision integrals which, for a hard disks
model, can be numerically evaluated. In particular, the bulk viscosity
is expressed in terms of the integral
\begin{equation}
\mathcal{I}\left(z\right)=\int_{\frac{2}{z}}^{\infty}e^{-\left(x-\frac{2}{z}\right)}
\left(\frac{1}{x}+\frac{3}{x^{2}}+\frac{3}{x^{3}}\right)\left(z^{2}x^{2}-4\right)^{5/2}dx\label{1}
\end{equation}
whose dependence on $z$ is not trivial. Moreover,
its convergence in the non-relativistic ($z\rightarrow0$) and ultrarelativistic
($z\rightarrow\infty$) limits is not completely justified. This is precisely part of
the goal of the present work, the formal establishment of the behaviour of Eq. (\ref{1})
in the non-relativistic and ultrarelativistic limits. Also, in order to address the effect
of a non-zero bulk viscosity of the gas for the complete range of $z$ we
study the corresponding modifications on the Rayleigh-Brillouin spectrum.

In order to accomplish such a task, the rest of the work is organized
as follows. In sections \ref{Ch-E} and \ref{visc} we briefly outline
the procedure carried out in Ref. \cite{Garcia-Mendez2018} in order
to establish the analytical expression for the bulk viscosity in the
specific case of a hard disks model. Section \ref{limits} is devoted
to the analysis of the integral defined in Eq. (\ref{1}) in the non
relativistic and ultrarelativistic limits, while section \ref{RB}
addresses the modification that a non-vanishing bulk viscosity has
on the Brillouin peaks for a light scattering spectrum. The discussion
of the results and final remarks are included in section \ref{final}.

\section{Boltzmann equation and Chapman-Enskog approximation}

\label{Ch-E}

In this section, we provide a brief description of the procedure leading to the
function $\mathcal{I}\left(z\right)$ in Eq. (\ref{1}), whose relevant
limits are explored in the main part of this work. The
complete calculation will be published elsewhere (a preprint can be
found in Ref. \cite{Garcia-Mendez2018}) together with the calculation
of the rest of the relevant coefficients, to which the reader is refered
to for further details.

The starting point is the relativistic Boltzmann equation \cite{CercignaniKremer}
in a flat spacetime with a ($+--$) metric which reads
\begin{equation}
v^{\alpha}\frac{\partial}{\partial x^{\alpha}}f\left(x^{\nu},v^{\nu}\right)=\int\int\left(\tilde{f}\tilde{f_{1}}-ff_{1}\right)F\sigma d\chi dv_{1}^{*},\label{2}
\end{equation}
where $f\left(x^{\nu},v^{\nu}\right)$ is the one particle distribution
function, $F$ is the invariant flux, $\gamma=u^{\mu}v_{\mu}/c^2$, with $v^{\mu}$ being the molecular
3-velocity measured in an
arbitrary frame, $u^{\mu}$ is the fluid's 3-velocity and $\sigma$ and $\chi$ are the
corresponding scattering cross section and solid angle in 2D. The
local equilibrium solution to Eq. (\ref{2}) is given by a bidimensional
Maxwell-J\"uttner distribution, that is

\begin{equation}
f^{\left(0\right)}\left(v^{\nu}\right)=\frac{ne^{\frac{1}{z}}}{2\pi c^{2}z\left(1+z\right)}e^{-\frac{u^{\alpha}v_{\alpha}}{zc^{2}}}\label{3}
\end{equation}
where $u^{\alpha}$ corresponds to the fluid's 3-velocity. With such
distribution one can establish statistical definition for the state
variables: $n$ (number density), $u^{\nu}$ (hydrodynamic 3-velocity)
and $\varepsilon$ (internal energy). Considering the particles frame
(Eckart's frame), one has
\begin{equation}
n=\int f^{\left(0\right)}\left(v^{\nu}\right)\gamma dv^{*},\label{4}
\end{equation}
\begin{equation}
nu^{\nu}=\int f^{\left(0\right)}\left(v^{\nu}\right)v^{\nu}dv^{*},\label{5}
\end{equation}
\begin{equation}
n\varepsilon=mc^{2}\int f^{\left(0\right)}\left(v^{\nu}\right)\gamma^{2}dv^{*},\label{6}
\end{equation}
from which one can obtain
\begin{equation}
n\varepsilon=nmc^{2}g\left(z\right)\qquad \textrm{with} \qquad g\left(z\right)=\frac{2z^{2}+2z+1}{z\left(z+1\right)}.\label{7}
\end{equation}
These quantities, in the absence of dissipation, follow Euler's equations.
In order to establish the first order in the gradients distribution,
the Chapman-Enkog solution to Eq. (\ref{2}) reads

\begin{equation}
f\left(v^{\nu}\right)=f^{\left(0\right)}\left(v^{\nu}\right)\left(1+\phi\left(v^{\nu}\right)\right),\label{8}
\end{equation}
where the first order correction to the local equilibrium distribution
function $\phi\left(v^{\nu}\right)$ is given by the solution of the
linearized Boltzmann equation.

\begin{equation}
f^{\left(0\right)}\int\int f^{\left(0\right)}\left(v_{1}\right)\left(\tilde{\phi}_{1}+\tilde{\phi}-\phi_{1}-\phi\right)F\sigma d\chi dv_{1}^{*}=v^{\alpha}\frac{\partial f^{\left(0\right)}\left(v^{\nu}\right)}{\partial x^{\alpha}}.\label{9}
\end{equation}
The dissipative fluxes, which arise from such deviation from equilibrium,
are found as moments of $f^{\left(0\right)}\left(v\right)\phi\left(v^{\nu}\right)$.
In particular, for the energy momentum tensor, which is the focus
of the present work, one has
\begin{equation}
\pi^{\mu\nu}=mh^{\mu\alpha}h^{\nu\beta}\int f^{\left(0\right)}\left(v\right)\phi\left(v^{\nu}\right)
v_{\alpha}v_{\beta}d^{*}v.\label{10}
\end{equation}
Here $h^{\alpha\beta}=\eta^{\alpha\beta}-u^{\alpha}u^{\beta}/c^{2}$
is the spatial projector corresponding to the (2+1) representation
where $u^{\alpha}$ corresponds to the temporal
direction in the comoving frame.

The details of the calculation that follows can be found in Ref.\cite{Garcia-Mendez2018}.
However the authors consider it worthwhile to point out here the particular
step to which the occurrence of a finite bulk viscosity in the relativistic
regime, opposed to the zero value obtained for non-relativistic gases,
can be traced down to. Equation (\ref{9}) is a linear integral equation,
whose solution is a superposition of the homogeneous and a particular
solutions. It is from the latter that the driving terms for the deviation
will appear as the gradients of the state variables. In such a term
one finds, for the part that depends on the velocity gradient and
will thus lead to viscous dissipation,
\begin{equation}
\left(v^{\alpha}f_{,\alpha}^{\left(0\right)}\right)_{visc.}=f^{\left(0\right)}\left\{ \left(\frac{v^{\beta}u_{\beta}}{c^{2}}\right)\left(\frac{1}{n}u^{\alpha}n_{,\alpha}+\frac{1}{T}\left(\frac{\gamma}{z}-g\left(z\right)\right)u^{\alpha}T_{,\alpha}\right)+\frac{v_{\mu}v^{\beta}h_{\beta}^{\alpha}}{zc^{2}}u_{,\alpha}^{\mu}\right\} \label{11}
\end{equation}
where we have introduced the usual decomposition $v^{\alpha}=v^{\beta}h_{\beta}^{\alpha}+\left(v^{\beta}u_{\beta}/c^{2}\right)u^{\alpha}$,
in this case in a (2+1) framework. Notice that the proper time derivatives
of the number density and temperature are coupled with the velocity
gradient through Euler's equations
\begin{equation}
u^{\alpha}n_{,\alpha}=-nu_{,\alpha}^{\alpha}\qquad\textrm{and}\qquad u^{\alpha}T_{,\alpha}=-k_{p}\left(z\right)Tu_{,\alpha}^{\alpha}\label{12-1}
\end{equation}
where $k_{p}\left(z\right)=k/C_{n}=\left(z+1\right)/\left(z
\left(g\left(z\right)+2\left(z+1\right)\right)\right)$.
Also, the term $u^{\alpha}u_{,\alpha}^{\mu}$ is not included in Eq.
(\ref{11}) since the momentum equation in the local equilibrium case
couples it only with the pressure gradient.

Introducing Eqs. (\ref{12-1}) in Eq. (\ref{11}) and splitting the
velocity gradient in antisymmetric, symmetric traceless parts and
the trace times an identity tensor one can write
\begin{equation}
\left(v^{\alpha}f_{,\alpha}^{\left(0\right)}\right)_{visc.}=\left[\left(\frac{1}{2}-k_{p}\left(z\right)\right)\frac{\gamma_{k}^{2}}{z}+\left(k_{p}
\left(z\right)g\left(z\right)-1\right)\gamma_{k}-\frac{1}{2z}\right]u_{,\nu}^{\nu}-\frac{v^{\mu}v^{\alpha}}{zc^{2}}\mathring{\sigma}_{\mu\alpha},
\label{13-1}
\end{equation}
where $\mathring{\sigma}_{\mu\alpha}$ is the symmetric traceless
component of $u_{\mu,\alpha}$. Notice that, in the non-relativistic
case, the term in brackets vanishes since in such limit $z\rightarrow0$,
$\gamma\rightarrow1$ and thus $g\left(z\right)\sim z^{-1}+1$ and
$k_{p}\left(z\right)\sim1-2z$. Due to this fact, the bulk viscosity
is zero for the monoatomic ideal gas at low temperatures. However,
as will be verified in the next section, is non-zero for $z\neq0$ and
increases rapidly for $0<z<1$.

Once the driving term given by Eq. (\ref{13-1}) is substituted in
Eq. (\ref{9}), a constitutive equation for the scalar part of the
Navier tensor can be found. Such a relation is written as
\begin{equation}
\pi_{\mu}^{\mu}=2\mu u_{,\alpha}^{\alpha}.\label{14-1}
\end{equation}

Following the standard procedure (see for example Refs. \cite{Ch-E,Courant})
one can obtain a first approximation for $\mu$. The steps are carefully
detailed in Ref. \cite{Garcia-Mendez2018}, here we only quote the
result:

\begin{equation}
\mu=\frac{4z^{7}}{\left(1+z\right)^{2}}\frac{mc^{2}}{\left(2z^{2}+4z+1\right)^{2}}\left[\gamma_{k}^{2},\gamma_{k}^{2}\right]^{-1},\label{15}
\end{equation}
where
\begin{equation}
\left[H,G\right]=-\frac{1}{n^{2}}\int\mathcal{C}\left(H\right)Gf^{\left(0\right)}d^{*}v\label{16}
\end{equation}
is the collision bracket, which satisfies
\begin{equation}
\left[H,G\right]=-\frac{1}{4n^{2}}\int\left(H'+H'_{1}-H-H_{1}\right)\left(G'+G'_{1}-G-G_{1}\right)f^{(0)}f_{*}^{(0)}F\sigma(\chi)d\chi dv_{1}^{*}dv^{*}d^{*}v.\label{17}
\end{equation}
In the next section, the behavior of $\mu$ as a function of $z$
for a hard disks model will be described.

\section{The temperature dependence of the viscosity for a hard disks gas}

\label{visc}

The so-called collision integrals defined in Eq. (\ref{16}) depend
on a molecular interaction model for the system. The simplest case
in a bidimensional scenario consists on a hard disk model for which the
scattering cross section is given by
\[
\sigma\left(\chi\right)=\frac{d}{2}\left|\sin\left(\frac{\chi}{2}\right)\right|.
\]
In such a case, as is shown in Appendix C of Ref. \cite{Garcia-Mendez2018},
the relevant collision integral is given by
\begin{equation}
\left[\gamma^{2},\gamma^{2}\right]=\frac{2cz^{3}d}{15}\frac{1}{\left(z+z^{2}\right)^{2}}\mathcal{I}\left(z\right),\label{18}
\end{equation}
where the integral $\mathcal{I}\left(z\right)$ is given in Eq. (\ref{1}).
Thus, one can write the bulk viscosity for the system as
\begin{equation}
\mu(z)=\frac{30mc}{d}\frac{z^{6}}{\left(2z^{2}+4z+1\right)^{2}}\mathcal{I}\left(z\right)^{-1}.\label{19}
\end{equation}
In the next section, the non-relativistic limit of this expression
will be carefully addressed. However, it can be seen at this point,
by inspection of Fig. \ref{fig:1} that $\mu$ vanishes at $z=0$. It is important to notice that for $z$ finite, $\mu\neq0$
and moreover, it increases very rapidly with temperature.

\begin{figure}
\begin{center}
\includegraphics[scale=0.8]{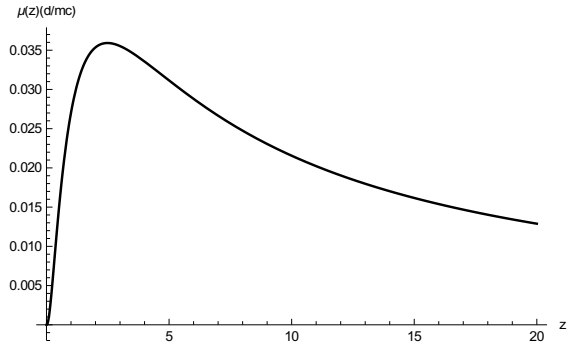}\caption{The dimensionless bulk viscosity as a function of the parameter $z$
in the relativistic scenario.\label{fig:1}}
\end{center}
\end{figure}
Bulk viscosity for the monoatomic ideal gas is thus non-zero as long
as $T\neq0$ and becomes relevant for some range of values of $z$.
The fact that it reaches a maximum value, which can be numerically
obtained as $\mu_{max}\sim\mu\left(2.4886\right)\sim0.0359$, and
then decreases for higher temperatures is also found in the three
dimensional case and deserves a closer analysis. One can then conclude
that for a limited range of temperature, for each gas, viscosity is
enhanced. This could lead to faster damping/enhancing of instabilities.

\section{Non-relativistic and ultrarelativistic limits of $\mu\left(z\right)$}

\label{limits}

In this section, a formal proof of the non-relativistic and ultrarelativistic
limits of the expression for $\mu$ obtained in Ref. \cite{Garcia-Mendez2018}
and quoted in the previous section is detailed. As mentioned above,
inspection of Fig. \ref{fig:1} points towards the bulk viscosity
approaching zero in the non-relativistic and ultra-relativistic limits.
This behavior is proven separately for each case. In particular, the
non-relativistic limit of this and other transport coefficients is
far from trivial since the lower integration limit tends to infinity
and thus the usual techniques cannot be applied.

We thus begin by addressing the low temperature, non-relativistic,
case. Let us start defining the function
\begin{equation}\label{funcionh}
h(x,z)=e^{-\left(x-\frac{2}{z}\right)}\left(\frac{1}{x}+\frac{3}{x^{2}}+\frac{3}{x^{3}}\right)\left(z^{2}x^{2}-4\right)^{5/2},
\end{equation}
which corresponds to the integrand in $\mathcal{I}\left(z\right)$ and an auxiliary function
\begin{equation}\label{funcionk}
k(x,z)=e^{-\left(x-\frac{2}{z}\right)}\left(\frac{1}{x^{2}}\right)\left(z^{2}x^{2}-4\right)^{3},
\end{equation}
both well defined on $A={(x,z) \in {\rm I\!R}^2, \mid xz \geq 2, z>0}$.

\begin{prop}\label{prop1}
   For all $(x,z)\in A$ and $z<1$, it is verified $0<k(x,z)\leq h(x,z)$.
\end{prop}

\begin{proof}
The equality holds for $x=2/z$. For $x>2/z$ we claim
\[
 e^{-\left(x-\frac{2}{z}\right)}\left(\frac{1}{x^{2}}\right)\left(z^{2}x^{2}-4\right)^{3}<
 e^{-\left(x-\frac{2}{z}\right)}\left(\frac{1}{x}+\frac{3}{x^{2}}+\frac{1}{x^{3}}\right)\left(z^{2}x^{2}-4\right)^{5/2}
\]
which is equivalent to
\[
  \left(z^{2}x^{2}-4\right)^{1/2}<\left(x+3+\frac{3}{x}\right)\,.
\]
If we define $x=2\alpha/z$ for $\alpha\in(1,\infty)$, last inequality can be written as
\[
   2z\left(\alpha^{2}-1\right)^{1/2}<2\alpha+3z+\frac{3z^{2}}{2\alpha}
\]
which holds for any $\alpha>1$ and for $z\in(0,1)$.
\end{proof}

\begin{theo}\label{teo1}
    $\mu(z)\rightarrow0^{+}$ as $z\rightarrow0^{+}$.
\end{theo}

\begin{proof}
Since $\mu\geq0$ for $z>0$, we can establish that the limit vanishes by simply upper-bounding the function by an auxiliary real-valued analytic function
that tends to zero in such limit.

From Proposition \ref{prop1} we have
\begin{equation}
0<\int_{\frac{2}{z}}^{\infty}k\left(x,z\right)dx\leq\int_{\frac{2}{z}}^{\infty}h\left(x,z\right)dx,\label{eq:inequality}
\end{equation}
where the integral on the left hand side can be computed as
\begin{eqnarray}
\int_{\frac{2}{z}}^{\infty}k\left(x,z\right)dx&=&8\left(-4z+2z^{2}-2z^{3}+3z^{4}+6z^{5}+3z^{6}\right)+64e^{2/z}\int_{2/z}^{\infty}\frac{e^{-t}}{t}dt,\\\nonumber
             &=& 96z^{5}+\mathcal{O}(z^{6})   \label{eq:k}
\end{eqnarray}
for $z\in\left(0,\epsilon\right)$, where $0<\epsilon<1$ is such that the last expression becomes true, namely there exists an $M>0$ and
\[
   \left|\int_{\frac{2}{z}}^{\infty}k\left(x,z\right)dx-96z^{5}\right|\leq M\left|z^{7}\right|
\]
for $0<z<\epsilon$.

Thus we have
\[
0<\mu(z)<\frac{30mc}{d}\frac{z^{6}}{\left(2z^{2}+4z+1\right)^{2}}\left( 96z^{5}+\mathcal{O}(z^{6})\right)^{-1},
\]
which formally shows that $\mu\rightarrow0^{+}$ as $z\rightarrow0^{+}$ in the non relativistic
limit.
\end{proof}

The ultrarelativistic case is established in the following.

\begin{theo}\label{teo2}
   $\mu(z)\rightarrow0^{+}$ as $z\rightarrow+\infty$.
\end{theo}

\begin{proof}
We realize that $\frac{2}{z}\rightarrow0$, therefore
\begin{eqnarray*}
\lim_{z\rightarrow\infty}\mathcal{I}\left(z\right)&=&\left(\lim _{z\rightarrow\infty}z^{5}\right)\int_{0}^{\infty}e^{-x}x^{2}\left(x^{2}+3x+3\right)dx\\
   &=&48\lim_{z\rightarrow\infty}z^{5}\,.
\end{eqnarray*}

Thus
\[
\lim_{z\rightarrow\infty}\mu(z)=\frac{5}{8}\frac{mc}{d}\lim_{z\rightarrow\infty}\frac{z}{\left(2z^{2}+4z+1\right)^{2}}=0.
\]
as is claimed
\end{proof}

Theorem \ref{teo1} and \ref{teo2} claim that the bulk viscosity
becomes negligible in the limits. This implies that, whatever effect it has on the dynamics
of the fluids, it should only be relevant in a finite interval of
$z$. In order to assess the possible effect of bulk viscosity dissipation
and provide an example of an experiment that may lead to the measure
of such effect, we carry out a linear analysis of a free relativistic
gas to first order in statistical density fluctuations.

\section{Modification to the Brillouin peaks due to bulk viscosity}

\label{RB}

The transport equation for the relativistic gas can be readily established
by multiplying Eq. (\ref{2}) by the collisional invariants and integrating
in velocity space. The procedure is the standard one, and leads to
two conservation equations
\begin{equation}
N_{,\nu}^{\nu}=0,\qquad T_{,\nu}^{\mu\nu}=0\label{20}
\end{equation}
where the particle flux is given by
\begin{equation}
N^{\nu}=\int f\left(v^{\nu}\right)v^{\nu}dv^{*}\label{21}
\end{equation}
and the energy momentum tensor is
\begin{equation}
T^{\mu\nu}=m\int f\left(v^{\nu}\right)v^{\mu}v^{\nu}dv^{*}.\label{22}
\end{equation}
The relation of such moments with the state variables is given, in
Eckart's frame and using the (2+1) decomposition by $N^{\nu}=nu^{\nu}$
and $T^{\mu\nu}=n\varepsilon u^{\mu}u^{\nu}/c^{2}+ph^{\mu\nu}+\pi^{\mu\nu}+q^{\mu}u^{\nu}/c^{2}+u^{\mu}q^{\nu}/c^{2}$.
Such expressions are then introduced in the balance equations (Eq.
(\ref{20}, \ref{21})) and the state variables are assumed to be given
by an equilibrium value plus a small fluctuation: $X=X_{0}+\delta X$.
The resulting system of equations, to first order in fluctuations
($\delta n$, $\delta u^{\mu}$ and $\delta T$ ) can be then transformed
to Fourier-Laplace space, in which the corresponding dispersion relation
is given by
\begin{equation}
s^{3}+a_{1}q^{2}s^{2}+\left(a_{2}q^{2}+a_{3}\right)q^{2}s+a_{4}q^{4}=0\label{23}
\end{equation}
where
\[
a_{1}=-\frac{1}{\rho}\left(A-\frac{\rho}{p}k_{p}\left(z\right)L_{T}+\frac{1}{c^{2}}\left(L_{n}+k_{p}\left(z\right)L_{T}\right)\right),
\]
\[
a_{2}=-\frac{k_{p}\left(z\right)}{\rho}AL_{T},
\]
\[
a_{3}=\frac{p}{\rho}\left(1+k_{p}\left(z\right)\right),
\]
\[
a_{4}=\frac{k_{p}\left(z\right)}{\rho}\left(L_{T}-L_{n}\right),
\]
here bulk viscosity enters in Eq. (\ref{23}) through
the relation $A=4\eta/3+\mu$ where $\eta$ is the shear viscosity,
and $L_{T}$ and $L_{n}$ are the transport coefficients appearing
in the relativistic heat flux constitutive equation \cite{Garcia-Mendez2018,Israel}. Equation (\ref{23})
has the same structure as the dispersion relation in the non-relativistic
case which can be analyzed using Mountain's method \cite{Mountain}. Following such
approximation, one can identify a purely decaying mode, corresponding
to a real root given by $s_{1}=-a_{4}q^{2}/a_{3}$. The remaining
two roots correspond to a conjugate pair
\begin{equation}
s_{2,\,3}=-\frac{1}{2}\left(a_{1}+\frac{a_{4}}{a_{3}}\right)\pm iq\sqrt{a_{3}}\label{24}
\end{equation}
which leads to decaying, oscillating modes. Thus, the corrections due
to the relativistic nature of the molecular dynamics of the disks could
be measured in a light scattering experiment where a Rayleigh-Brillouin
spectrum can be obtained. In particular, since the focus of this work
is the effect of the bulk viscosity, we are only concerned with the
Brillouin doublet width, which is given by the real part of the complex
roots (Eq. (\ref{24})). The width of the central peak as well as
the location of the lateral ones is not affected by the presence of
$\mu$. Let's call
\begin{equation}
W\left(\mu(z)\right)=-\frac{1}{2}\left(a_{1}+\frac{a_{4}}{a_{3}}\right),\label{25}
\end{equation}
the width of the Brillouin peaks. Also, in order to understand better the effect of $\mu$, in Fig.
\ref{fig:2}  we show the ratio $\left(W(\mu)-W(\mu=0)\right)/W(\mu=0)$ as a function of $z$. Clearly
$W(\mu \neq 0)$ and $W(\mu =0)$
are equal at $z=0$, the ratio reaches a maximum for an intermediate
value of $z$ and finally tends to zero for large values of the temperature.
\begin{figure}
\begin{center}
    \includegraphics[scale=0.8]{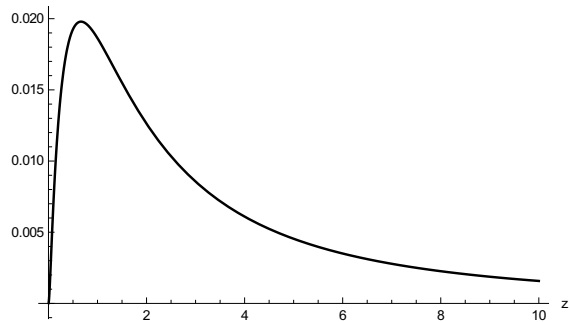}\caption{\label{fig:2} The ratio of the Brillouin widths
    $\left(W(\mu)-W(\mu=0)\right)/W(\mu=0)$ as a function of $z$.}
\end{center}
\end{figure}

\section{Discussion and final remarks}

\label{final} The integral in Eq. (\ref{1}) that drives to the complete
identification of the dependence of $\mu$ with the temperature in
the two-dimensional relativistic system is far of being simple. In
this work we showed that in both the non-relativistic and ultrarelativistic
limits the bulk viscosity of a 2D relativistic system vanishes and
the formal proof was presented.

Dissipation in single component fluids in the absence of external
forces is composed of two effects: viscous and thermal. While thermal
dissipation in relativistic systems modifies in a somewhat dramatical
fashion the structure of the transport equations, the viscous tensor
effects remain the same in form however the coefficients are significantly
altered. In particular, the Rayleigh-Brillouin spectrum is modified
by the presence of a non-vanishing bulk viscosity in the relativistic
scenario. The dynamics of density fluctuations clearly depend on all dissipative
contributions. In general, the shape of the Brillouin peaks is given
in terms of all the transport coefficients (see Eq. (\ref{25})) but
the presence of a non zero bulk coefficient in the intermediate temperature
regime leads to a slight modification of the spectrum, in particular
the width of the lateral peaks is altered, as is shown in Fig. \ref{fig:2}
and Eq. (\ref{25}). The analysis of the extreme limits shown in this
paper for all the transport coefficients in the (2+1) case is important
in the context of the development of new two-dimensional materials
as graphene and will be addressed elsewhere.

\end{document}